%% file: limavg_arxiv.tex
\pdfoutput=1
\documentclass[sigconf, authorversion, nonacm, screen]{acmart}

\usepackage{graphicx}
\usepackage{textcomp}
\usepackage{mathtools}
\usepackage{color,soul}
\usepackage{url}
\usepackage{algpseudocode}
\usepackage{algorithm}
\usepackage{outlines}

\usepackage{tikz}
\input{tikzsettings.tex}

\AtBeginDocument{%
	\providecommand\BibTeX{{%
			\normalfont B\kern-0.5em{\scshape i\kern-0.25em b}\kern-0.8em\TeX}}}

\newtheorem{assum}{Assumption}
\newtheorem{defn}{Definition}
\newtheorem{rem}{Remark}
\newtheorem{prop}{Proposition}

\newtheorem{thm}{Theorem}

\input{notation.tex}

\newcommand{\fakeparagraphnospace}[1]{\vspace{1mm}\noindent\textbf{#1.}}
\newcommand{\fakeparagraph}[1]{\fakeparagraphnospace{#1}}

\begin{document}
	
	\title{Computing the sampling performance of event-triggered control}
	
	\author{Gabriel de A.~Gleizer}
	\email{g.gleizer@tudelft.nl}
	\affiliation{%
		\institution{TU Delft}
		\streetaddress{Mekelweg}
		\city{Delft}
		\country{The Netherlands}}
	
	\author{Manuel Mazo Jr.}
	\email{m.mazo@tudelft.nl}
	\affiliation{%
		\institution{TU Delft}
		\streetaddress{Mekelweg}
		\city{Delft}
		\country{The Netherlands}}
	
	\begin{abstract}
		
		In the context of networked control systems, event-triggered control (ETC) has emerged as a major topic due to its alleged resource usage reduction capabilities. However, this is mainly supported by numerical simulations, and very little is formally known about the traffic generated by ETC.
		This work devises a method to estimate, and in some cases to determine exactly, the minimum average inter-sample time (MAIST) generated by periodic event-triggered control (PETC) of linear systems. 
		The method involves abstracting the traffic model using a bisimulation refinement algorithm and finding the cycle of minimum average length in the graph associated to it. This always gives a lower bound to the actual MAIST. Moreover, if this cycle turns out to be related to a periodic solution of the closed-loop PETC system, the performance metric is exact.
	\end{abstract}

	\maketitle

	\section{INTRODUCTION}
	%
	Nowadays, most control systems are implemented using networks as the communication medium between sensors, controllers, and actuators. The classic method of performing control through digital media is called periodic \emph{sample-and-hold control}, where sensor data is gathered at a fixed sampling rate, control commands are then immediately updated and sent to actuators, which hold this command for the sampling period time. Choosing the sampling period is based fundamentally on a worst-case analysis across the state-space of the system. Disrupting this periodic paradigm, event-triggered control (ETC) works by sampling only when a significant condition happens, thus adapting the sampling rate to the system state: this gives it the potential to drastically reduce communications in comparison to periodic control. This concept dates back from \cite{aastrom2002comparison}, with the formal methods to design event conditions to achieve desired stability properties presented in \cite{tabuada2007event}. Since then, many authors (e.g.~\cite{wang2008event, girard2015dynamic, heemels2012introduction}) worked on event design to improve sampling performance while guaranteeing stability and control performance properties, or to improve the practical implementation aspects of ETC, such as the periodic event-triggered control (PETC) of \cite{heemels2013periodic}. In PETC, event conditions are checked periodically, but sensor values are only communicated upon the condition satisfaction.
	It is important to note, however, that most of the evidence of superiority of ETC in comparison to periodic control, in what concerns bandwidth usage, is only supported through numerical case studies. In some cases, such as PETC, one can obtain a straightforward \emph{qualitative} assertion of non-inferiority by setting the \emph{event-checking period} equal to the largest periodic sampling time one can obtain. Nonetheless, no \emph{quantitative} measurement of this superiority has been established; e.g., in PETC, doing this strategy often leads to periodic triggering of the events, bringing no benefits at all. 
	
	The present work is concerned precisely with measuring the sampling performance of a given PETC system. More specifically, we aim at computing the \emph{minimum average inter-sample time} (MAIST) of a PETC system: this translates directly to its expected network load or resource utilization. We focus on PETC due to its practical relevance, but also because, as observed in \cite{gleizer2020scalable}, it enables exact traffic abstractions---a central tool for the results of this paper.
	Literature related to this objective can be categorized in two main approaches. 
	The first \cite{postoyan2019interevent} focuses on understanding the qualitative asymptotic trends of the inter-sample times of planar linear systems. The authors conclude that, for small enough triggering condition parameters, the inter-sample times eventually converge to a fixed value or exhibit a periodic pattern. Despite providing very interesting insights, the results are limited to two-dimensional state-spaces, do not provide the quantitative information we are interested in, and, perhaps most importantly, are only valid for small triggering parameters: this way, ETC provides the least benefit. The second category is the use of symbolic abstractions \cite{kolarijani2016formal, gleizer2020scalable}, which follow on the extensive work on partitioning and aggregation for abstractions, see \cite{tabuada2009verification}. These works are concerned with short-term prediction of inter-sample times in order to develop a scheduler that can, e.g., request sensor data before events are triggered; they do not capture long-term properties of the sampling behavior of ETC like the MAIST, which, we argue, provide a more definitive information about the sampling performance. Still in the same category, \cite{gleizer2020towards} has recently given a step towards understanding longer-term traffic patterns of PETC, by proposing the usage of a bisimulation-like algorithm which determines the $m$ next inter-sample times from a given state. This allows a very conservative estimate of the MAIST by taking the minimum average of all such $m$-length sequences.
	
	The present work builds upon the bisimulation-like algorithm of \cite{gleizer2020towards}, which can be seen as a modified $l$-complete abstraction \cite{moor1999supervisory, schmuck2015comparing}, to compute the MAIST of PETC for linear time-invariant systems. The main insight is that computing the MAIST once there is a finite-state simulation (abstraction) of the PETC traffic is easy, as it reduces to finding the cycle of minimum average length of the associated weighted graph \cite{chatterjee2010quantitative}. %
	Then, as we show, if we can \emph{verify} that the minimum cycle in the abstraction exists in the actual closed-loop system, we will have obtained an exact value for the MAIST. This observation gives rise to the concept of \emph{minimum-average-cycle-equivalent simulation}. If the minimum cycle is not exhibited by the PETC system, the value obtained through the abstraction is still a lower bound, and further refinements will eventually break it, providing tighter bounds.
	
	\subsection{Notation}
	
	We denote by $\No$ the set of natural numbers including zero, $\N \coloneqq \No \setminus \{0\}$, $\N_{\leq n} \coloneqq \{1,2,...,n\}$, by $\Q$ the set of rational numbers, and by $\R_+$ the set of non-negative reals. %
	We denote by $|\xv|$ the norm of a vector $\xv \in \R^n$, but if $s$ is a sequence or set, $|s|$ denotes its length or cardinality, respectively. For a square matrix $\Am \in \R^{n \times n},$ we write $\Am \succ \O$ ($\Am \succeq \O$) if $\Am$ is positive definite (semi-definite). The set $\S^n$ denotes the set of symmetric matrices in $\R^n$. 
	For a set $\Xs\subseteq\Omega$, we denote by $\bar{\Xs}$ its complement: $\Omega \setminus \Xs$; 
	We often use a string notation for sequences, e.g., $\sigma = abc$ reads $\sigma(1) = a, \sigma(2) = b, \sigma(3) = c.$ Powers and concatenations work as expected, e.g., $\sigma^2 = \sigma\sigma = abcabc.$ In particular, $\sigma^\omega$ denotes the infinite repetition of $\sigma$.
	For a relation $\Rs \subseteq \Xs_a \times \Xs_b$, its inverse is denoted as $\Rs^{-1} = \{(x_b, x_a) \in \Xs_b \times \Xs_a : (x_a, x_b) \in \Rs\}$. Finally, we denote by $\pi_\Rs(\Xs_a) \coloneqq \{ x_b \in \Xs_b \mid (x_a, x_b) \in \Rs \text{ for some } x_a \in \Xs_b\}$ the natural projection of $\Xs_a$ onto $\Xs_b$.
	
	\section{PROBLEM STATEMENT}
	
	Consider a linear time-invariant plant controlled with sample-and-hold state feedback \cite{astrom2008event} described by
	\begin{align}
		\dot{\xiv}(t) &= \Am\xiv(t) + \Bm\Km\hat{\xiv}(t),\label{eq:plant}
	\end{align}
	where $\xiv(t) \in \R^\nx$ is the plant's state with initial value $\xv_0 \coloneqq \xiv(0)$, $\hat{\xiv}(t) \in \R^\nx$ is the state measurement available to the controller, $\Km\hat{\xiv}(t) \in \R^\nup$ is the control input, $\nx$ and $\nup$ are the state-space and input-space dimensions, respectively, and $\Am, \Bm, \Km$ are matrices of appropriate dimensions. 
	The holding mechanism is zero-order: let $t_i \in \R_+, i \in \N_0$ be a sequence of sampling times, with $t_0 = 0$ and $t_{i+1} - t_i > \varepsilon$ for some $\varepsilon > 0$; then $\hat{\xiv}(t) = \xiv(t_i), \forall t \in [t_i, t_{i+1})$. 
	
	In ETC, a \emph{triggering condition} determines the sequence of times $t_i$. In PETC, this condition is checked only periodically, with a fundamental checking period $h$. Figure \ref{fig:block} shows a simple diagram depicting the ETC scheme. We consider the family of static \emph{quadratic triggering conditions} from \cite{heemels2013periodic} with an additional maximum inter-event time condition below:
	\begin{equation}\label{eq:quadtrig}
		t_{i+1} = \inf\left\{\!kh>t_i, k \in \N ~\middle|
		\!\begin{array}{c}
			\begin{bmatrix}\xiv(kh) \\ \xiv(t_i)\end{bmatrix}\tran
			\!\Qm \begin{bmatrix}\xiv(kh) \\ \xiv(t_i)\end{bmatrix} > 0\! \\
			\text{ or } \ kh-t_i \geq \bar{k}h\phantom{\dot{\hat{I}}}
		\end{array}\!
		\right\}\!,
	\end{equation}
	where $\Qm \in \S^{2\nx}$ is the designed triggering matrix, and $\bar{k}$ is the chosen maximum (discrete) inter-event time.\footnote{Typically, a maximum inter-event time exists naturally for a system with (P)ETC (see \cite{gleizer2018selftriggered}). Still, one may want to set a smaller maximum inter-event time so as to establish a ``heart beat'' of the system. In any case, this is a necessity if one wants to obtain a finite-state abstraction of the system.} 
	Observing this equation, note that the inter-event time $t_{i+1} - t_{i}$ is a function of $\xv_i \coloneqq \xiv(t_i)$; denoting $\kappa \coloneqq (t_{i+1}-t_i)/h$ as the discrete inter-sample time, it follows that
	\begin{gather}
		\kappa(\xv_i) = \min\left\{k \in \{1, 2, ...\bar{k}\} \mid \xv_i\tran\Nm(k)\xv_i > 0 \text{ or } k=\bar{k}\right\}, \nonumber\\
		\Nm(k) \coloneqq \begin{bmatrix}\Mm(k) \\ \I\end{bmatrix}\tran
		\Qm \begin{bmatrix}\Mm(k) \\ \I\end{bmatrix}, \label{eq:petc_time}\\
		\!\!\Mm(k) \coloneqq \Am_\d(k) + \Bm_\d(k)\Km \coloneqq \e^{\Am hk} + \int_0^{hk}\e^{\Am\tau}\d\tau \Bm\Km.\!\!\nonumber
	\end{gather}
	where $\I$ denotes the identity matrix. Thus, the event-driven evolution of sampled states can be compactly described by the recurrence
	\begin{equation}\label{eq:samples}
		\xiv(t_{i+1}) = \Mm(\kappa(\xiv(t_i))\xiv(t_i).
	\end{equation}
	With this, each initial condition $\xv_0 \in \R^\nx$ leads to a sequence of samples $\xv_i$ and inter-sample times $k_i(\xv_0)$ defined recursively as 
	\begin{align*}
		\xv_{i+1} &= \Mm(\kappa(\xv_i))\xv_i \\
		k_i(\xv_0) &\coloneqq \kappa(\xv_i),
	\end{align*}
	for which one may attribute an \emph{average inter-sample time} (AIST):
	$$ \text{AIST}(\xv) \coloneqq \liminf_{n\to\infty}\frac{1}{n+1}\sum_{i=0}^{n}hk_i(\xv). $$
	As usual, we use $\liminf$ instead of $\lim$ to obtain the limit lower bound in case the regular limit does not exist.
	
	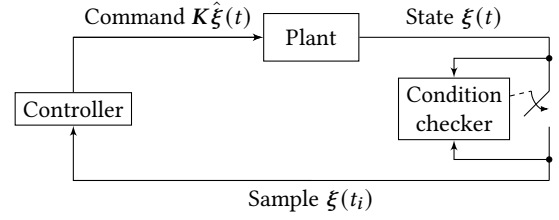
\begin{figure}
		\begin{center}
			\input{block_diagram.tex}
			\caption{\label{fig:block} Block diagram of an ETC system.}
		\end{center}
	\end{figure}
	
	The goal of this paper is to devise a method to compute, for a given periodic event-triggered controlled system \eqref{eq:plant}--\eqref{eq:quadtrig}, its \emph{minimum average inter-sample time} (MAIST), which is the minimal AIST across all possible initial conditions:
	\begin{equation}\label{eq:maist}
		\text{MAIST} \coloneqq \inf_{\xv \in \R^\nx}\liminf_{n\to\infty}\frac{1}{n+1}\sum_{i=0}^{n}hk_i(\xv).
	\end{equation}
	
	Calculating or even estimating the quantity above is challenging. How can one choose a sufficiently large $n$, or how can one exhaustively search for states to obtain one that yields the MAIST? This direct approach is unpromising, and thus we propose to find the value of Eq.~\eqref{eq:maist} through finite-state abstractions.
	
	\section{BACKGROUND}
	
	The strategy to solve the problem posed in the previous section is to \emph{abstract} the infinite-state system given by Eqs.~\eqref{eq:plant}--\eqref{eq:quadtrig} to a finite-state system, compute the equivalent to a MAIST in this abstraction, and establish a relation between the quantities of the original system and its abstraction. For that, we introduce the framework of \cite{tabuada2009verification} to formally relate systems of different natures, e.g., those described by differential equations with those described by finite-state machines. Later, we present the notion of quantitative automata from \cite{chatterjee2010quantitative} and how to compute the MAIST of a (priced) automaton.
	
	\subsection{Transition systems and abstractions}
	
	In \cite{tabuada2009verification}, Tabuada gives a generalized notion of transition system:
	\begin{defn}[Transition System \cite{tabuada2009verification}]\label{def:system} 
		A system $\Ss$ is a tuple $(\Xs,\Xs_0,\Es,\Ys,H)$ where:
		\begin{itemize}
			\item $\Xs$ is the set of states,
			\item $\Xs_0 \subseteq \Xs$ is the set of initial states,
			\item $\Es \subseteq \Xs \times \Xs$ is the set of edges (or transitions),
			\item $\Ys$ is the set of outputs, and
			\item $H: \Xs \to \Ys$ is the output map.
		\end{itemize}
	\end{defn}
	Here we have omitted the action set $\Us$ from the original definition because we focus on autonomous systems. A system is said to be finite (infinite) state when the cardinality of $\Xs$ is finite (infinite). A transition in $\Es$ is denoted by a pair $(x, x')$. We define $\Post_\Ss(x) \coloneqq \{x'\mid (x,x') \in \Es\}$ as the set of states that can be reached from $x$ in one step. System $\Ss$ is said to be \emph{non-blocking} if $\forall x \in \Xs, \Post_\Ss(x) \neq \emptyset.$ 
	We call $x_0x_1x_2...$ an \emph{infinite internal behavior}, or \emph{run} of $\Ss$ if $x_0 \in \Xs_0$ and $(x_i,x_{i+1}) \in \Es$ for all $i \in \N$, and $y_0y_1...$ its corresponding \emph{infinite external behavior}, or \emph{trace}, if $H(x_i) = y_i$ for all $i \in \N$. We denote by $B_{\Ss}(r)$ the external behavior from a run $r = x_0x_1...$ (in the case above, $B_{\Ss}(r) = y_0y_1...$), by $\Bs^\omega_x(\Ss)$ the set of all infinite external behaviors of $\Ss$ starting from state $x$, and by $\Bs^\omega(\Ss) \coloneqq \bigcup_{x\in\Xs_0}\Bs^\omega_x(\Ss)$ the set of all infinite external behaviors of $\Ss$.
	
	The concepts of simulation and bisimulation are fundamental to establish formal relations between two transition systems.
	
	\begin{defn}[Simulation Relation \cite{tabuada2009verification}]\label{def:sim}
		Consider two systems $\Ss_a$ and $\Ss_b$ with $\Ys_a$ = $\Ys_b$. A relation $\Rs \subseteq \Xs_a \times \Xs_b$ is a simulation relation from $\Ss_a$ to $\Ss_b$ if the following conditions are satisfied:
		\begin{enumerate}
			\item[i)] for every $x_{a0} \in \Xs_{a0}$, there exists $x_{b0} \in \Xs_{b0}$ with $(x_{a0}, x_{b0}) \in \Rs;$
			\item[ii)] for every $(x_a, x_b) \in \Rs, H_a(x_a) = H_b(x_b);$
			\item[iii)] for every $(x_a, x_b) \in \Rs,$ we have that $(x_a, x_a') \in \Es_a$ implies the existence of $(x_b, x_b') \in \Es_b$ satisfying $(x_a', x_b') \in \Rs.$
		\end{enumerate}
	\end{defn}
	We say $\Ss_a \preceq \Ss_b$ when $\Ss_b$ simulates $\Ss_a$, which is true if there exists a simulation relation from $\Ss_a$ to $\Ss_b$. 
	\begin{defn}[Bisimulation \cite{tabuada2009verification}]
		Consider two systems $\Ss_a$ and $\Ss_b$ with $\Ys_a$ = $\Ys_b$.
		$\Ss_a$ is said to be bisimilar to $\Ss_b$, denoted $\Ss_a \bisim \Ss_b$, if there exists a relation
		$\Rs$ such that:
		\begin{itemize}
			\item $\Rs$ is a simulation relation from $\Ss_a$ to $\Ss_b$;
			\item $\Rs^{-1}$ is a simulation relation from $\Ss_b$ to $\Ss_a$.
		\end{itemize}
	\end{defn}
	Weaker but relevant relations associated with simulation and bisimulation are, respectively, \emph{behavioral inclusion} and \emph{behavioral equivalence}:
	\begin{defn}[Behavioral inclusion and equivalence \cite{tabuada2009verification}]
		Consider two systems $\Ss_a$ and $\Ss_b$ with $\Ys_a$ = $\Ys_b$. We say that $\Ss_a$ is \emph{behaviorally included} in $\Ss_b$, denoted by $\Ss_a \preceq_\Bs \Ss_b$, if $\Bs^\omega(\Ss_a) \subseteq \Bs^\omega(\Ss_b).$ In case $\Bs^\omega(\Ss_a) = \Bs^\omega(\Ss_b),$ we say that $\Ss_a$ and $\Ss_b$ are \emph{behaviorally equivalent}, which is denoted by $\Ss_a \bisim_\Bs \Ss_b$.
	\end{defn}
	(Bi)simulations lead to behavioral inclusion (equivalence):
	\begin{thm}[\cite{tabuada2009verification}]
		Given two systems $\Ss_a$ and $\Ss_b$ with $\Ys_a$ = $\Ys_b$:
		\begin{itemize}
			\item $\Ss_a \preceq \Ss_b \implies \Ss_a \preceq_\Bs \Ss_b$;
			\item $\Ss_a \bisim \Ss_b \implies \Ss_a \bisim_\Bs \Ss_b$.
		\end{itemize}
	\end{thm}
	
	\subsection{Quantitative automata}
	
	While much of the field of formal methods in control is concerned with qualitative analyses, such as safety, stability, and reachability, often quantitative computations are of interest, like the problem we set ourselves to solve in this paper. In \cite{chatterjee2010quantitative}, Chatterjee et al.~established a comprehensive framework for quantitative problems on finite-state systems, from which we borrow some definitions and results, trying to be consistent with the notation from the previous section as much as possible.
	
	\begin{defn}[Weighted automaton (adapted from \cite{chatterjee2010quantitative})]
		A weighted automaton $\Ss$ is the tuple $(\Xs,\Xs_0,\Es,\Ys,H, \gamma)$, where
		\begin{itemize}
			\item $(\Xs,\Xs_0,\Es,\Ys,H)$ is a \emph{non-blocking} transition system;
			\item $\gamma: \Es \to \Q$ is the \emph{weight function}.
		\end{itemize}
	\end{defn}
	The adaptation we have made is that we include outputs to comply with previously introduced notation; again, we ignore the action set as we are interested in autonomous systems. For a given run $r = x_0x_1...$ of $\Ss$, $\gamma(r) = v_0v_1...$ is the sequence of weights defined by $v_i = \gamma(x_i,x_{i+1})$.
	
	A \emph{value function} $\text{Val} : \Q^\omega \to \R$ attributes a value to an infinite sequence of weights $v_0v_1... \in \Q^\omega$. Among the well-studied value functions, the one we are interested in is
	$$ \LimAvg(v) \coloneqq \liminf_{n\to\infty}\frac{1}{n+1}\sum_{i=0}^{n}v_i.$$
	A \emph{LimAvg-automaton} is a weighted automaton equipped with the LimAvg value function. We define the value of a LimAvg-automaton as $V(\Ss) \coloneqq \inf\{\LimAvg(\gamma(r)) \mid r \text{ is a run of } \Ss\}.$\footnote{In \cite{chatterjee2010quantitative}, $\sup$ is used instead of $\inf$ because they consider it a more natural choice in general quantitative decision problems, considering the convention used in qualitative decisions (e.g., acceptance). As it is remarked in \cite{chatterjee2010quantitative} itself, using $\inf$ is also a valid choice, and it is the most natural for our problem.} The following result is essentially an excerpt from Theorem 3 in \cite{chatterjee2010quantitative}, which uses the classical result from Karp \cite{karp1978characterization}:
	\begin{thm}\label{thm:limavg}
		Given a finite-state LimAvg-automaton $\Ss$ with $|\Xs| = n$ and $|\Es| = m, V(\Ss)$ can be computed in $\bigO(nm)$. Moreover, system $\Ss$ admits a cycle $x_0x_1...x_k$ satisfying $x_i \to x_{i+1}, i < k,$ and $x_k \to x_0$ s.t. $\LimAvg(\gamma((x_0x_1...x_k)^\omega)) = V(\Ss).$
	\end{thm}
	The cycle mentioned above is a \emph{minimum average cycle} of the weighted digraph defined by $\Ss$, and can be recovered in $\bigO(n)$ using the algorithm in \cite{chaturvedi2017note}.
	
	\section{COMPUTING THE MAIST}
	
	From Theorem \ref{thm:limavg}, we have an indication that it would be relatively straightforward to compute the minimum average inter-sample time of the PETC system \eqref{eq:plant}--\eqref{eq:quadtrig} if we could represent it as a weighted automaton. Let us investigate how we can do this. 	We start by describing the evolution of samples of a PETC system, cf.~Eq.~\eqref{eq:samples}, as a generalized transition system following Def.~\ref{def:system}:
	\begin{equation}\label{eq:S}
		\begin{aligned}
			\Ss = (\R^n&, \R^n, \Es, \Ys, H), \text{ where} \\ 
			\Es & = \{(\xv,\xv') \in \R^n \times \R^n \mid \xv' = \Mm(\kappa(\xv))\xv\} \\
			\Ys &= \{1,2,...,\bar{k}\} \\ 
			H &= \kappa.
		\end{aligned}
	\end{equation}
	The first feature that we see by inspecting Eqs.~\eqref{eq:S} and \eqref{eq:maist} in view of the definition of a LimAvg-automaton is that the weight of a transition is in fact $h$ times the output of its outbound state. Hence, for any run $r$ of $\Ss$, it holds that $\gamma(r) = h \cdot B_\Ss(r)$; that is, we can characterize weight sequences, hence run values, exclusively by external behaviors. Considering the notion of behavioral inclusion, this gives a straightforward result: 
	\begin{prop}\label{prop:bound}
		Consider two systems $\Ss_a$ and $\Ss_b$ with $\Ys_a = \Ys_b \subset \Q$. Attribute to each system the weight function $\gamma_s(x_s, x'_s)$ $\equiv$ $ H_s(x_s)$, where $s\!\in\! \{a, b\}.$ If $\Ss_a \preceq_\Bs\!\!(\bisim_\Bs)\, \Ss_b$, then $V(\Ss_a) \,\geq(=)\, V(\Ss_b).$ 
	\end{prop}
	\begin{proof} By definition, $V(\Ss_s) = \inf\{\LimAvg(\gamma_s(r)) \mid r $ is a run of $\Ss_s\} = \inf\{\LimAvg(y) \mid y \in \Bs^\omega(\Ss_s)\}.$ Since $\Bs^\omega(\Ss_a) \subseteq\!\!(=)$ $\Bs^\omega(\Ss_b)$, the desired result follows.
	\end{proof}
	Proposition \ref{prop:bound} hints that obtaining a finite-state (bi)simulation of Eq.~\ref{eq:S} provides means to compute a lower bound (or the actual value) for the MAIST. This is promising, because works in \cite{gleizer2020scalable, gleizer2020towards} provides methods to find simulations of a PETC traffic. However, on the one hand, a simulation alone does not provide how conservative the lower bound may be; on the other hand, a finite-state bisimulation of an infinite system is often impossible to be obtained. In fact, bisimulation is too strong, in the sense that all behaviors and their fragments are exactly captured. As hinted by Theorem \ref{thm:limavg}, the LimAvg value is determined by a minimum average cycle of the system. If one such cycle happens to have a correspondence with the concrete system, this is sufficient to obtain the exact value for the MAIST. For the rest of this paper, we assume every transition system is equipped with the weight function equal to the outbound state output, i.e., $\gamma(x, x') \equiv  H(x)$. Let us proceed with formalities.
	\begin{defn}[Minimum-average-cycle-equivalent simulation] \label{def:macesim} Consider two transition systems $\Ss_a$ and $\Ss_b$ satisfying $\Ss_a \preceq \Ss_b$. Denote by $\MAC(\Ss_b)$ the set of minimum average cycles of $\Ss_b$. If $dc^\omega \in \Bs^\omega(\Ss_a)$ for some finite-length sequence $d$ and some $c \in \MAC(\Ss_b)$, then $\Ss_b$ is a minimum-average-cycle-equivalent (MACE) simulation of $\Ss_a$.
	\end{defn}
	A MACE simulation is a standard simulation with the additional requirement that one of the minimum average cycles of the abstraction must be observed on the concrete system, after possibly a finite number of transitions from the initial state. It should be clear that MACE simulation is stronger than simulation, 
	but it is significantly weaker than bisimulation. The following result is a straightforward conclusion from Proposition \ref{prop:bound} and Theorem \ref{thm:limavg}.
	\begin{prop}
		Let $\Ss_b$ be a finite-state MACE simulation of $\Ss_a$; then, $V(\Ss_a) = V(\Ss_b)$.
	\end{prop}
	\begin{proof}
		From Def.~\ref{def:macesim}, take $dc^\omega \in \Bs^\omega(\Ss_a)$ for some finite-length sequence $d$ with $c = k_0k_1...k_N \in \MAC(\Ss_b)$. The associated $\LimAvg$ is
		$$ v \coloneqq \LimAvg(dc^\omega) = \LimAvg(c^\omega) = \frac{1}{N}\sum_1^Nk_i. $$
		From Theorem \ref{thm:limavg}, $v = V(\Ss_b)$. As $v$ is also the value of a behavior from $\Ss_a$, it holds that $v \geq V(\Ss_a)$. Since Prop.~\ref{prop:bound} gives that $V(\Ss_a) \geq V(\Ss_b)$, we have that
		$$ V(\Ss_b) = v \geq V(\Ss_a) \geq V(\Ss_b), $$
		and thus $V(\Ss_b) = V(\Ss_a)$.
	\end{proof}
	
	\subsection{MACE simulation of PETC traffic}
	
	The challenge now resides on obtaining a MACE simulation of the system \eqref{eq:S}. For this we need to be able to (i) build a finite-state simulation of the system; (ii) check if its minimum mean cycle exists in the actual system; if not, (iii) refine the simulation until the cycle breaks; and (iv) repeat the process. This method is essentially the same as the bisimulation algorithm from a quotient model, presented in \cite{tabuada2009verification}, which was used for PETC in \cite{gleizer2020towards}, but with a different stopping criterion. Therefore, let us recover the simulation relation in \cite{gleizer2020towards}, with a simplification that suits our purpose:
	\begin{defn}[Inter-sample sequence relation (adapted from \cite{gleizer2020towards})]\label{def:bisimrel} Given a sequence length $l$, we denote by $\Rs_l \subseteq \Xs \times \Ys^l$ the relation satisfying 
		$(\xv,\dummy{k_1k_2...k_l}) \in \Rs_l$ if and only if
		\begin{subequations}\label{eq:sequence}
			\begin{align}
				\xv &\in \Qs_{k_1}, \label{eq:sequence_k1}\\
				\Mm(k_1)\xv &\in \Qs_{k_2}, \label{eq:sequence_k2}\\
				\Mm(k_2)\Mm(k_1)\xv &\in \Qs_{k_3}, \label{eq:sequence_k3}\\
				& \vdots \nonumber\\
				\Mm(k_{l-1})...\Mm(k_1)\xv &\in \Qs_{k_l}, \label{eq:sequence_kl-1}
			\end{align}
		\end{subequations}
		where 
		\begin{equation}\label{eq:setq}
			\begin{gathered}
				\Qs_k \coloneqq \Ks_k \setminus \left(\bigcap_{j=\underline{k}}^{k-1} \Ks_{j}\right) = \Ks_k \cap \bigcap_{j=1}^{k-1} \bar{\Ks}_{j}, \\
				\Ks_k \coloneqq \begin{cases}
					\{\xv \in \Xs| \xv\tran\Nm(k)\xv > 0\}, & k < \bar{k}, \\
					\R^\nx, & k = \bar{k}.
				\end{cases}
			\end{gathered}
		\end{equation}
	\end{defn}
	
	Eq.~\eqref{eq:setq}, taken from \cite{gleizer2020scalable}, defines the sets $\Qs_k$, containing the states that trigger exactly with inter-sample time $hk$. Eq.~\eqref{eq:sequence} simply states that a state $\xv \in \R^n$ is related to a state $k_1k_2...k_l$ of the abstraction if the inter-sample time sequence that it generates for the next $l$ samples is $hk_1,hk_2,...,kh_l$. The simplification with respect to \cite{gleizer2020towards} is that, here, we are not concerned with the state reaching a ball around the origin, but only with the sequence of sampling times it generates.
	
	\begin{rem} Setting $l=1$ gives a quotient state set of Eq.~\ref{eq:S}, while larger values of $l$ can be seen as refinements using the bisimulation algorithm of \cite{tabuada2009verification}. This relation can also be seen as a method to construct the strongest $l$-complete approximation \cite{moor1999supervisory} of the system \eqref{eq:S}. As a consequence of \cite[Corollary 2]{schmuck2015comparing}, for autonomous deterministic systems, both methods lead to the same abstraction. Hence, here we can use both terms interchangeably, but for purposes of exposition we will use the term \emph{$l$-complete} in what follows.
	\end{rem}
	\begin{defn}\label{def:lsim} Given an integer $l \geq 1$, the \emph{$l$-complete PETC traffic model} is the system $\Ss_l \coloneqq \left(\Xs_l, \Xs_l, \Es_l, \Ys, H_l\right)$, with 
	\begin{itemize}
		\item $\Xs_l \coloneqq \pi_{\Rs_l}(\Xs)$,
		\item $\Es_l = \{(k\sigma, \sigma k') \mid k,k' \in \Ys, \sigma \in \Ys^{l-1}, k\sigma, \sigma k' \in \Xs_l\},$
		\item $H_l(k_1k_2...k_m) = k_1.$
	\end{itemize}
	\end{defn}
	The model above partitions the state-space $\R^\nx$ of the PETC into subsets associated with the next $l$ inter-sample times these states generate. Computing the state set requires solving the quadratic inequality satisfaction problems of Eq.~\eqref{eq:sequence}, which can be done exactly using a satisfiability-modulo-theories (SMT) solver\footnote{For that, the variable is $\xv \in \R^\nx$ and the query is $\exists \xv: $Eq.~\eqref{eq:sequence} holds.} such as Z3 \cite{demoura2008z3}, or approximately through convex relaxations as proposed in \cite{gleizer2020scalable}.\footnote{Using relaxations implies finding inter-sample sequences that may not be exhibited by the real system. This still generates a simulation relation, but not the strongest $l$-complete approximation.} The output map is the next sample alone. The transition relation is simply what is called in \cite{schmuck2015comparing} the domino rule: a state associated with a sequence $k_1k_2...k_l$ must naturally lead to a state whose next first $l-1$ samples are $k_2k_3...k_l$, because the system is deterministic, autonomous, and time-invariant. Hence, any state in $\Xs_l$ that starts with $k_2k_3...k_l$ is a possible successor of $k_1k_2...k_l$. An example for $l=1,2,3$ is depicted in Fig.~\ref{fig:Sl}.
	
	The following fact is a direct consequence of Theorems 6 and 7 from \cite{schmuck2015comparing}, and gives the desired simulation refinement properties:
	\begin{prop}
		Consider the system $\Ss$ from Eq.~\eqref{eq:S} and $\Ss_l$ from Definition \ref{def:lsim}, for some $l \geq 1$. Then, $\Ss \preceq \Ss_{l+1} \preceq \Ss_l.$
	\end{prop}
	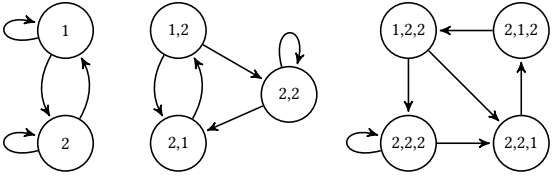
\begin{figure}[tb]
		\begin{center}
			\begin{footnotesize}
				\begin{tikzpicture}[->,>=stealth',shorten >=1pt, auto, node distance=1.5cm,
					semithick]
					\tikzset{every state/.style={minimum size=3em, inner sep=2pt}}
					
					\node[state] 		 (1)    {1};
					\node[state]         (2) [below of=1] {2};
					
					\path (1) edge[bend right] (2) (2) edge[bend right] (1)
					(1) edge [loop left] (1)
					(2) edge [loop left] (2);
					
					\node[state] 		 (12) [right of=1]  {1,2};
					\node[state]		 (21) [below of=12] {2,1};
					\node[state]		 (22) at ([shift=({-30:1.7 cm})]12) {2,2};
					
					\path (12) edge (22) (22) edge [loop above] (22) (22) edge (21) (21) edge[bend right] (12) (12) edge[bend right] (21);
					
					\node[state] 		 (122) [right=2.3cm of 12]  {1,2,2};
					\node[state] 		 (212) [right of=122]     {2,1,2};
					\node[state] 		 (222) [below of=122]     {2,2,2};
					\node[state] 		 (221) [right of=222]     {2,2,1};
					
					\path (221) edge (212) (212) edge (122) (122) edge (222) edge (221) (222) edge[loop left] (222) (222) edge (221);

				\end{tikzpicture}
			\end{footnotesize}
			\caption{\label{fig:Sl} Example of $l$-complete PETC traffic models, for $l=1$ (left), $l=2$ (middle), and $l=3$ (right).}
			\vspace{-1.5em}
		\end{center}
	\end{figure}
	
	\fakeparagraph{Periodic set solutions of a PETC} 
	The result above allows us to construct finite-state abstractions of system \eqref{eq:S} with increased precision by increasing $l$; for each $l$, Karp's algorithm \cite{karp1978characterization} (with the corrections from \cite{chaturvedi2017note}) can be used to detect the mean average cycles. What is missing is a method to verify any such cycle in the concrete system. To this end, we need to confirm whether there is a subset of $\R^n$ that is a periodic solution with the desired pattern. Given that the recursion \eqref{eq:samples} is piecewise linear, and due to the homogeneity of the sets $\Q_k$ of Eq.~\eqref{eq:setq}, linear subspaces suffice to characterize these periodic solutions in the general case.\footnote{A set $\Qs$ is homogeneous if $\xv \in \Qs \implies c\xv \in \Qs, \ \forall c \in \R \setminus \{0\}.$ In the general case, every solution of a linear recurrence $\xiv(k+1) = \Mm\xiv(k)$ converges asymptotically to an invariant linear subspace of $\Mm$, hence verifying these invariants suffice. A formal proof is left out due to space limitations.}
	
	\begin{defn}[Periodic subspace solution] 
	A linear subspace $\As \subseteq \R^\nx$ is a \emph{periodic subspace solution} of Eq.~\eqref{eq:samples} if there exists $J\in\N$ such that, for every $\xv \in \As \setminus \{\emph{\O}\}$, $\xiv(t_i) = \xv \implies \xiv(t_{i+J}) \in \As$. 
	\end{defn}
	\begin{rem}
		We remove the origin from the desired set solution (hence abusing the term \emph{linear} subspace) because it is the only point along a ray that does not trigger like the others, due to the strict inequality in Eq.~\eqref{eq:quadtrig}. Since controllers are typically designed for asymptotic stability, solutions starting away from the origin never reach it. If $\xiv(0) = \O$, it will always trigger at $\bar{k}$, so $\bar{k}^\omega$ is an obvious behavior, although very unlikely to be the minimum-in-average.
	\end{rem}
	\begin{prop}\label{prop:subsol}
		A linear subspace $\As$ is a periodic subspace solution of Eq.~\eqref{eq:samples} if there exists a sequence $k_1, k_2,..., k_J$ such that (i) $\As$ is an invariant of $\Mm(k_J)...\Mm(k_2)\Mm(k_1)$, and (ii)
		\begin{equation}\label{eq:As_squence}
		\begin{aligned}
			\As \setminus \{\O\} &\subseteq \Qs_{k_1},\\
			\Mm(k_1)\As \setminus \{\O\} &\subseteq \Qs_{k_2},\\
			& \vdots \\
			\Mm(k_{J-1})...\Mm(k_1)\As \setminus \{\O\} &\subseteq \Qs_{k_J}.
		\end{aligned}
		\end{equation}
	\end{prop}
	\begin{proof}
		First, $\As$ must be a fixed set of Eq.~\eqref{eq:samples} iterated $J$ times, hence condition (i). Second, for a particular sequence $k_1, k_2,..., k_J$, all points in this set must satisfy Eq.~\eqref{eq:sequence}, which is what is displayed in Eq.~\eqref{eq:As_squence}.
	\end{proof}

	Proposition \ref{prop:subsol} shows how one can find if a sequence associated with the minimum average cycle of a similar model is indeed a periodic solution of the PETC system. First, determine the invariants of $\Mm(k_J)...\Mm(k_2)\Mm(k_1)$, then, verify if they satisfy Eq.~\ref{eq:As_squence}. For this latter part, first remember that each $\Qs_k$ is an intersection of quadratic sets (see Eq.~\eqref{eq:setq}). Then, the following simple result can be used to check if the linear space is a subset of a given quadratic set:
	\begin{prop}\label{prop:subspaceposdef}
		Let $\As$ be a linear subspace with basis $\vv_1, \vv_2,..., \vv_m$, and let $\Vm$ be the matrix composed of the vectors $\vv_i$ as columns. Let $\Qm \in \S^n$ be a symmetric matrix and define $\Qs_n \coloneqq \{\xv \in \R^n \mid  \xv\tran\Qm\xv \geq 0\}$ and $\Qs_s \coloneqq \{\xv \in \R^n \mid  \xv\tran\Qm\xv > 0\}$. Then, $\As \setminus \{\O\} \subseteq \Qs_n$ (resp.~$\Qs_s$) if and only if $\Vm\tran\Qm\Vm \succeq \O$ (resp.~$\Vm\tran\Qm\Vm \succ \O$).
	\end{prop}
	\begin{proof}
		For brevity, let us consider the strict inequality case (the other is analogous). First, note that $\As = \{\Vm\av \mid \av \in \R^m\}.$ Hence, if we want all points in $\As$ to belong to $\Qs_s$, we need that $$ \forall \av \in \R^m \setminus \{\O\}, \ \av\tran\Vm\tran\Qm\Vm\av > 0,$$ which is exactly the definition of $\Vm\tran\Qm\Vm \succ \O$.
	\end{proof}
	To have a finite collection of subspaces to be checked, we rely on the following assumption:
	\begin{assum}\label{assum:diffeigs}
		Let $k_1,k_2,...,k_J \in \MAC(\Ss_l)$. Then, the matrix $\Mm(k_J)...\Mm(k_2)\Mm(k_1)$ has no repeated eigenvalues.
	\end{assum}
	\begin{rem} If Assumption \ref{assum:diffeigs} holds, then $\Mm(k_J)...\Mm(k_2)\Mm(k_1)$ has at most $2^\nx$ real invariant subspaces, which are the linear combinations of its eigenvectors.\footnote{The special case of repeated eigenvalues require some technicalities and is left out of this paper due to space considerations. } \end{rem}
	
	\fakeparagraph{The MACE simulation algorithm} Algorithm \ref{alg} summarizes the method to obtain a MACE simulation of a given PETC system \eqref{eq:plant}--\eqref{eq:quadtrig}. In the outer loop, the relation $\Rs_l$ and corresponding finite-state system $\Ss_l$ are built, followed by the computation of one of its minimum average cycles. Then, an inner loop looks for linear subspaces associated with this cycle that satisfies Prop.~\ref{prop:subsol}; if one is found, the algorithm terminates. Otherwise, $l$ is incremented and the main loop is repeated. An important remark is that, if a cycle $\sigma^\omega$ does not exist in the concrete system, after sufficiently many iterations it will not be exhibited in the abstraction, as for some $N$, $\sigma^N$ will not be a sequence satisfying Eq.~\eqref{eq:sequence}.
	\begin{algorithm}\caption{\label{alg}MACE simulation algorithm}
		\begin{flushleft}
		\hspace*{\algorithmicindent} \textbf{Input:} $h, \Ys$, $\ \Mm(k), \Qs_k, \forall k \in \Ys$ \\
		\hspace*{\algorithmicindent} \textbf{Output:} $l, \Ss_l, \mathtt{value}, \mathtt{cycle}, \mathtt{MAIST}$
		\end{flushleft}
		\begin{algorithmic}[1]
			\State $l \gets 1$
			\While{true}
				\State Build $\Rs_l$ and $\Ss_l$ \Comment{(Defs.~\ref{def:bisimrel} and \ref{def:lsim})}
				\State $\mathtt{value} \gets V(\Ss_l), \ \mathtt{cycle} \gets \MAC(\Ss_l)$ \Comment{\cite{karp1978characterization, chaturvedi2017note}}
					\State $\Vs \gets \text{eigenvecs}(\Mm_{k_m}...\Mm_{k_2}\Mm_{k_1})$
					\Comment{$k_1k_2...k_m = \mathtt{cycle}$}
					\For{$\Vs' \in 2^\Vs \setminus \{\emptyset\}$}
						\State $\As \gets \text{span}(\Vs')$
						\If{$\As$ satisfies Prop.~\ref{prop:subsol} with $k_1,k_2,...,k_m$}
							\State $\mathtt{MAIST} \gets h \cdot \mathtt{value}$
							\State \Return
						\EndIf
					\EndFor
				\State $l \gets l+1$
			\EndWhile
		\end{algorithmic}
	\end{algorithm}

	For an example, refer to Fig.~\ref{fig:Sl}, and assume the cycle $(1,2,2)^\omega$ is a trace of the concrete system. For $l=1$, the only MAC is $1^\omega$, but it is not verified in the PETC. For $l=2$, this cycle is broken, and the MAC becomes $(1,2)^\omega$, again unverified. Finally, for $l=3$, the MAC is $(1,2,2)^\omega$, which is verified and the algorithm terminates.
	
	\begin{rem}\label{rem:semi}
		Algorithm \ref{alg} is in fact a semi-algorithm: it terminates if the concrete system exhibits a periodic trace whose LimAvg value is the smallest among all of its solutions. This is not always the case with PETC, as it may exhibit non-periodic traces. E.g., consider system \eqref{eq:plant}--\eqref{eq:samples} with $\bar{k} = 2, \Mm(1) = \Mm(2) = \alpha\begin{bsmallmatrix}1 & 2 \\ -2 & 1\end{bsmallmatrix}, \Nm(1) = \begin{bsmallmatrix}0 & 1 \\ 1 & 0\end{bsmallmatrix}.$ Matrices $\Mm(k)$ have eigenvalues equal to $\alpha(1 \pm 2\mathrm{i})$, which means that solutions rotate (up to rescaling) on the plane by the irrational angle $\arctan(2)$. $\Nm(1)$ selects the first and third quadrants to output 1, whereas the others output 2. Identifying $\xv \sim \lambda \xv, \ \forall \lambda \in \R \setminus \{0\}$, this system is topologically equivalent to an irrational rotation on the circle, which does not exhibit periodic behavior.
	\end{rem}
	
	\section{Numerical example}
	
	Consider the system \eqref{eq:plant} with
	$$ \Am = \begin{bmatrix}0 & 1 \\ -2 & 3 \end{bmatrix}, \Bm = \begin{bmatrix}0 \\ 1\end{bmatrix}, \Km =\begin{bmatrix}0 & -5\end{bmatrix}, $$
	and the triggering condition of \cite{tabuada2007event}, $|\xiv(t) - \hat{\xiv}(t)| > \sigma|\xiv(t)|$ for some $0 < \sigma < 1$, which can be put in the form Eq.~\eqref{eq:quadtrig}. Checking time was set to $h=0.05$, and maximum inter-sample time to $\bar{k}=20$. We implemented Algorithm \ref{alg} in Python using Z3 \cite{demoura2008z3} to solve Eq.~\eqref{eq:sequence}, and attempted to compute its MAIST through a MACE simulation for $\sigma \in \{0.1, 0.2, 0.3, 0.4, 0.5\}$. Table \ref{tab} presents the MAIST for each $\sigma$, as well as the $l$ value (Def.~\ref{def:lsim}) where it was obtained. Only for $\sigma = 0.1$ the algorithm did not terminate before $l=50$: for this case, the actual $\bar{k}$ of the system was 3, and all $\Mm(k), k \leq 3,$ have complex eigenvalues. Thus, it is likely that it does not have periodic behaviors, similarly to what is discussed in Remark \ref{rem:semi}.\footnote{It cannot be easily proven, however, whether there is no admissible product of such matrices with real eigenvalues.} Nonetheless, the LimAvg value of the \emph{maximum} average cycle of $\Ss_l$ is 0.0798 (after multiplying by $h$); since this is an upper bound for the MAIST (same arguments as Prop.~\ref{prop:bound}), we know that the estimate is within only 0.0012 of the real value. For the other cases, trivial cycles were found for $\sigma = 0.4 (5^\omega)$ and $\sigma = 0.5 (6^\omega)$, but it took a few iterations to break, e.g., the $2^\omega$ loop. Interestingly, the simplest cycles for $\sigma = 0.2$ and $\sigma = 0.3$ had length, respectively, 27 and 28, showing that PETC can often lead to very complex recurring patterns.
	
	\begin{table}\caption{\label{tab} MAIST values for the numerical example}
	\begin{tabular}{c|ccccc}
		\hline
		$\sigma$ & 0.1 & 0.2 & 0.3 & 0.4 & 0.5 \\
		\hline
		$l$ & 50* & 15 & 26 & 12 & 10 \\
		MAIST & 0.0786 & 0.137 & 0.171 & 0.25 & 0.3 \\
		CPU time [s] & 327 & 41 & 147 & 29 & 45 \\
		\hline
	\end{tabular}

	\begin{flushleft}
		{\footnotesize * Algorithm interrupted before finding a verified cycle.}
	\end{flushleft}
		\vspace{-1em}
	\end{table}
	
	\section{CONCLUSIONS}
	
	We have presented a method to compute the sampling performance of PETC, namely its minimum average inter-sample time, by means of abstractions. For that, we observed that limit average metrics of finite-state priced automata are (easily) computable; then, using behavioral inclusion properties of abstractions, we introduced the concept of minimum-average-cycle-equivalent simulations, together with a semi-algorithm that can compute the MAIST of linear PETC systems whenever their limit behaviors are periodic. When behaviors are aperiodic, the algorithm still provides good approximations, increasingly with higher values of $l$. 
	
	Future work is aimed at how to use these methods to design better triggering mechanisms particularly for sampling performance and schedulability. As we have noted in the previous section, computing the maximum AIST is also possible with the same tools; since a simulation model has embedded in it which states lead which cycles, one could consider modifying the sampling strategy for the states that lead to the MAIST, in order to steer them towards the system's most sample-efficient behaviors.	%
	
	\begin{acks}
		This work is supported by the \grantsponsor{GSERC}{European Research Council}{https://erc.europa.eu/} through the SENTIENT project, Grant No.~\grantnum[https://cordis.europa.eu/project/id/755953]{GSERC}{ERC-2017-STG \#755953}.
	\end{acks}
	
	\bibliographystyle{ieeetr} 
	\bibliography{mybib} 
		
\end{document}

%% file: tikzsettings.tex
\usepackage{pgfplots}
\usepackage{grffile}
\pgfplotsset{compat=newest}
\usetikzlibrary{automata}
\usetikzlibrary{arrows, arrows.meta, decorations.markings, patterns}
\usepgfplotslibrary{patchplots, groupplots}
\usetikzlibrary{shapes, positioning, fit, backgrounds}
\usepgfplotslibrary{fillbetween}

\usetikzlibrary{overlay-beamer-styles}
\usetikzlibrary{calc}

\usepackage{calc}

\definecolor{tudcyan}{RGB}{0,166,214}
\definecolor{tudmagenta}{RGB}{109,23,127}
\definecolor{tudpurple}{RGB}{29,28,115}
\definecolor{tudgraygreen}{RGB}{107,134,137}

\colorlet{lighttudcyan}{tudcyan!20}
\colorlet{lighttudmagenta}{tudmagenta!20}

\newlength{\hatchspread}
\newlength{\hatchthickness}
\newlength{\hatchshift}
\newcommand{\hatchcolor}{}
\tikzset{hatchspread/.code={\setlength{\hatchspread}{#1}},
	hatchthickness/.code={\setlength{\hatchthickness}{#1}},
	hatchshift/.code={\setlength{\hatchshift}{#1}},
	hatchcolor/.code={\renewcommand{\hatchcolor}{#1}}}
\tikzset{hatchspread=7pt,
	hatchthickness=0.5pt,
	hatchshift=0pt,
	hatchcolor=black}
\pgfdeclarepatternformonly[\hatchspread,\hatchthickness,\hatchshift,\hatchcolor]
{custom north west lines}
{\pgfqpoint{\dimexpr-2\hatchthickness}{\dimexpr-2\hatchthickness}}
{\pgfqpoint{\dimexpr\hatchspread+2\hatchthickness}{\dimexpr\hatchspread+2\hatchthickness}}
{\pgfqpoint{\dimexpr\hatchspread}{\dimexpr\hatchspread}}
{
	\pgfsetlinewidth{\hatchthickness}
	\pgfpathmoveto{\pgfqpoint{0pt}{\dimexpr\hatchspread+\hatchshift}}
	\pgfpathlineto{\pgfqpoint{\dimexpr\hatchspread+0.15pt+\hatchshift}{-0.15pt}}
	\ifdim \hatchshift > 0pt
	\pgfpathmoveto{\pgfqpoint{0pt}{\hatchshift}}
	\pgfpathlineto{\pgfqpoint{\dimexpr0.15pt+\hatchshift}{-0.15pt}}
	\fi
	\pgfsetstrokecolor{\hatchcolor}
	\pgfusepath{stroke}
}

\def\centerarc[#1](#2)(#3:#4:#5)
{ \draw[#1] ($(#2)+({#5*cos(#3)},{#5*sin(#3)})$) arc (#3:#4:#5); }

%% file: notation.tex
\newcommand*{\tran}{^{\mkern-1.5mu\mathsf{T}}\!}  

\def\d{\ensuremath{\mathrm{d}}}

\DeclareMathOperator{\Post}{Post}
\DeclareMathOperator{\LimAvg}{LimAvg}
\DeclareMathOperator{\MAC}{MAC}

\DeclareMathOperator{\bigO}{\mathcal{O}}

\def\norm[#1]{\left|#1\right|}
\def\shortnorm[#1]{|#1|}
\def\bisim{\ensuremath{\cong}}

\def\dummy{}

\def\Xs{\mathcal{X}}
\def\Ys{\mathcal{Y}}
\def\Us{\mathcal{U}}

\def\Vs{\mathcal{V}}
\def\No{\mathbb{N}_{0}}
\def\N{\mathbb{N}}
\def\R{\mathbb{R}}
\def\S{\mathbb{S}}
\def\Q{\mathbb{Q}}

\def\Rs{\mathcal{R}}

\def\Ss{\mathcal{S}}
\def\Ks{\mathcal{K}}
\def\Qs{\mathcal{Q}}

\def\As{\mathcal{A}}  
\def\Es{\mathcal{E}}  
\def\Bs{\mathcal{B}}  
\def\Us{\mathcal{U}}  

\def\nup{{n_{\mathrm{u}}}}
\def\nx{{n_{\mathrm{x}}}}

\def\xv{\boldsymbol{x}}

\def\vv{\boldsymbol{v}}
\def\xiv{\boldsymbol{\xi}}

\def\av{\boldsymbol{a}}

\def\Am{\boldsymbol{A}}
\def\Bm{\boldsymbol{B}}

\def\I{{\normalfont\textbf{I}}}
\def\O{{\normalfont\textbf{0}}}

\def\Mm{\boldsymbol{M}}
\def\Nm{\boldsymbol{N}}

\def\Qm{\boldsymbol{Q}}

\def\Vm{\boldsymbol{V}}

\def\Km{\boldsymbol{K}}

\def\e{\mathrm{e}}



%% file: block_diagram.tex
\tikzstyle{block} 		= [draw, rectangle, minimum height=2em, minimum width=4em]
\tikzstyle{sum} 		= [draw, circle, inner sep=0, minimum size=0.2cm, node distance=1cm]
\tikzstyle{input} 		= [coordinate]
\tikzstyle{output} 		= [coordinate]
\tikzstyle{split} 		= [coordinate]
\tikzstyle{pinstyle} 	= [pin edge={to-,thin,black}]

\begin{tikzpicture}[auto, node distance=1.5em,>=latex']

\node [block, align=center,] (P) {Plant};


\node [coordinate, below = 5 em of P] (C) {};


\node[coordinate, right = 10 em of C.east](br){};
\node[coordinate, output, left = 10 em of C.west](bl){};
\draw (P.east -| br) node[coordinate] (ur){};
\draw (P.west -| bl) node[coordinate] (ul){};

\draw [-] (P.east) --node[above]{State $\xiv(t)$} (ur);
\draw [-] (br) --node[below]{Sample $\xiv(t_i)$} (bl);
\draw [->] (ul) --node[above]{Command $\Km\hat{\xiv}(t)$} (P.west);

\path (ur) -- coordinate[midway](mr) (br);
\path (ul) -- node[rectangle, midway, anchor=center, draw](ml){Controller} (bl);
\draw [->] (C) -| (ml);
\draw [->] (ml) |- (P);

\node[coordinate, above = 0.75 em of mr] (rbegs){};
\node[coordinate, below = 1.5 em of rbegs](rends){};
\node[coordinate, below left = 1.5 em of rbegs](rups){};
\draw [-] (ur) -- (rbegs);
\draw [-] (rbegs) -- (rups);
\draw [-] (rends) -- (br);
\node[coordinate, left = 0.75 em of rbegs](rarrs){};
\node[coordinate, below = 0.00 em of mr] (rarre){};
\draw [-{Latex[length=0.35em,width=0.25em]}] (rarrs) to [out=-90, in=180, looseness=1] (rarre);

\node[coordinate, above = 0.00 em of ml] (larre){};

\node[coordinate, left = 1. em of mr] (sr) {};
\node[coordinate, right = 1. em of ml] (sl) {};
\node[circle, inner sep = 0, minimum size=2pt, draw, fill=black, below = 0.75 em of ur] (wiretop) {};
\node[circle, inner sep = 0, minimum size=2pt, draw, fill=black, above = 0.75 em of br] (wirebot) {};
\node[left = 4 em of mr, anchor=center, text centered, draw, solid, text width=4em](T){Condition checker} (C);
\draw[->] (wiretop) -| (T);
\draw[->] (wirebot) -| (T);
\draw[dash pattern={on 0.2em off 0.2em}] (T) -- (rarrs);
%
\end{tikzpicture}

%% file: limavg_arxiv.bbl
\begin{thebibliography}{10}

\bibitem{aastrom2002comparison}
K.~J. {\AA}str{\"o}m and B.~Bernhardsson, ``Comparison of riemann and lebesgue
  sampling for first order stochastic systems,'' in {\em Proceedings of the
  41st IEEE Conference on Decision and Control, 2002}, vol.~2, pp.~2011--2016,
  IEEE, 2002.

\bibitem{tabuada2007event}
P.~Tabuada, ``Event-triggered real-time scheduling of stabilizing control
  tasks,'' {\em IEEE Transactions on Automatic Control}, vol.~52, no.~9,
  pp.~1680--1685, 2007.

\bibitem{wang2008event}
X.~Wang and M.~D. Lemmon, ``Event design in event-triggered feedback control
  systems,'' in {\em Decision and Control, 2008. CDC 2008. 47th IEEE Conference
  on}, pp.~2105--2110, IEEE, 2008.

\bibitem{girard2015dynamic}
A.~Girard, ``Dynamic triggering mechanisms for event-triggered control,'' {\em
  IEEE Transactions on Automatic Control}, vol.~60, no.~7, pp.~1992--1997,
  2015.

\bibitem{heemels2012introduction}
W.~Heemels, K.~H. Johansson, and P.~Tabuada, ``An introduction to
  event-triggered and self-triggered control,'' in {\em Decision and Control
  (CDC), 2012 IEEE 51st Annual Conference on}, pp.~3270--3285, IEEE, 2012.

\bibitem{heemels2013periodic}
W.~P. M.~H. Heemels, M.~C.~F. Donkers, and A.~R. Teel, ``Periodic
  event-triggered control for linear systems,'' {\em IEEE Transactions on
  Automatic Control}, vol.~58, no.~4, pp.~847--861, 2013.

\bibitem{gleizer2020scalable}
G.~A. Gleizer and M.~Mazo~Jr., ``Scalable traffic models for scheduling of
  linear periodic event-triggered controllers,'' {\em 21st IFAC World Congress
  (accepted)}, 2020.
\newblock \url{https://arxiv.org/abs/2003.07642}.

\bibitem{postoyan2019interevent}
R.~Postoyan, R.~G. Sanfelice, and W.~P. M.~H. Heemels, ``Inter-event times
  analysis for planar linear event-triggered controlled systems,'' in {\em
  Decision and Control, 2019. CDC 2019. 58th IEEE Conference on}, 2019.

\bibitem{kolarijani2016formal}
A.~S. Kolarijani and M.~Mazo~Jr, ``A formal traffic characterization of {LTI}
  event-triggered control systems,'' {\em IEEE Transactions on Control of
  Network Systems}, 2016.

\bibitem{tabuada2009verification}
P.~Tabuada, {\em Verification and control of hybrid systems: a symbolic
  approach}.
\newblock Springer Science \& Business Media, 2009.

\bibitem{gleizer2020towards}
G.~A. Gleizer and M.~Mazo~Jr., ``Towards traffic bisimulation of linear
  periodic event-triggered controllers,'' {\em IEEE Control Systems Letters},
  vol.~5, no.~1, pp.~25--30, 2021.

\bibitem{moor1999supervisory}
T.~Moor and J.~Raisch, ``Supervisory control of hybrid systems within a
  behavioural framework,'' {\em Systems \& control letters}, vol.~38, no.~3,
  pp.~157--166, 1999.

\bibitem{schmuck2015comparing}
A.-K. Schmuck, P.~Tabuada, and J.~Raisch, ``Comparing asynchronous l-complete
  approximations and quotient based abstractions,'' in {\em 2015 54th IEEE
  Conference on Decision and Control (CDC)}, pp.~6823--6829, IEEE, 2015.

\bibitem{chatterjee2010quantitative}
K.~Chatterjee, L.~Doyen, and T.~A. Henzinger, ``Quantitative languages,'' {\em
  ACM Transactions on Computational Logic (TOCL)}, vol.~11, no.~4, pp.~1--38,
  2010.

\bibitem{astrom2008event}
K.~J. Astr{\"o}m, ``Event based control,'' in {\em Analysis and design of
  nonlinear control systems}, pp.~127--147, Springer, 2008.

\bibitem{gleizer2018selftriggered}
G.~A. Gleizer and M.~Mazo~Jr., ``Self-triggered output feedback control for
  perturbed linear systems,'' {\em IFAC-PapersOnLine}, vol.~51, no.~23,
  pp.~248--253, 2018.

\bibitem{karp1978characterization}
R.~M. Karp, ``A characterization of the minimum cycle mean in a digraph,'' {\em
  Discrete mathematics}, vol.~23, no.~3, pp.~309--311, 1978.

\bibitem{chaturvedi2017note}
M.~Chaturvedi and R.~M. McConnell, ``A note on finding minimum mean cycle,''
  {\em Information Processing Letters}, vol.~127, pp.~21--22, 2017.

\bibitem{demoura2008z3}
L.~De~Moura and N.~Bj{\o}rner, ``{Z3}: An efficient {SMT} solver,'' in {\em
  International conference on Tools and Algorithms for the Construction and
  Analysis of Systems}, pp.~337--340, Springer, 2008.

\end{thebibliography}
